\newif\ifnotes
\notestrue 

\documentclass[conference]{IEEEtran}
\IEEEoverridecommandlockouts
\usepackage{amsmath,amssymb,amsfonts,amsthm}
\usepackage{mathtools}

\newtheorem{corollary}{Corollary}

\usepackage{graphicx}
\usepackage{textcomp}
\usepackage{xcolor}
\usepackage{flushend}
\usepackage[caption = false]{subfig}

\newcommand{\solname}{\textsc{Nectar}\xspace}
\newcommand{\mesure}{data sent\xspace}

\newcommand{\editft}[1]{#1}
\newcommand{\editms}[1]{#1}

\newcommand{\mschange}[1]{\ifnotes{\color{black}#1}\else#1\fi}
\newcommand{\ftchange}[1]{\ifnotes{\color{black}#1}\else#1\fi}

\usepackage[
  colorlinks=true,
  linkcolor=purple,
  citecolor=purple,
  pdfauthor={},
  pdftitle={},
  pdfsubject={},
  pdfkeywords={},
  bookmarks=false,
]{hyperref}

\usepackage[
  sorting=none,
  style=numeric-comp,
  backend=biber,
  isbn=false,
  url=true,
  doi=true,
  url=true,
  mincrossrefs=100,
  maxnames=12,
    doi=false, 
    url=false,
 giveninits
]{biblatex}

\usepackage{tikz}
\usepackage{soul}

\usetikzlibrary{math}
\usetikzlibrary{arrows.meta}
\usetikzlibrary{patterns}
\usetikzlibrary{calc}
\usetikzlibrary{decorations.pathreplacing}
\usetikzlibrary{shapes.callouts}

\usepackage[linesnumbered,vlined,boxed]{algorithm2e}

\SetInd{0.6em}{0.4em}
\SetKwComment{Comment}{$\rhd$ }{}
\SetCommentSty{textit}
\SetKwInOut{Input}{Input}\SetKwInOut{Output}{Output}

\SetKwProg{WhenReceived}{when}{ is received do}{}
\SetKwProg{Round}{In each synchronous round}{ do}{end round}
\SetKwIF{If}{ElseIf}{Else}{if}{then}{else if}{else}{}

\newcommand{\notPartitionable}{\textsc{Not\_Partitionable}\xspace}
\newcommand{\partitionable}{\textsc{Partitionable}\xspace}
\newcommand{\confirmed}{\ensuremath{\mathsf{confirmed}}\xspace}
\newcommand{\true}{\ensuremath{\mathsf{True}}\xspace}
\newcommand{\false}{\ensuremath{\mathsf{False}}\xspace}

\newcommand{\send}{\ensuremath{\mathsf{send}}\xspace}
\newcommand{\sendto}{\ensuremath{\mathsf{to}}\xspace}

\newcommand{\proofWord}{\ensuremath{\mathit{proof}}\xspace}
\newcommand{\neighbor}{\ensuremath{\mathit{neighbor}}\xspace}

\newcommand{\lengthSign}{\ensuremath{\mathsf{lengthSign}}\xspace}

\newcommand{\reachable}{\ensuremath{\mathsf{DetectReachableNode}}\xspace}
\newcommand{\connectivity}{\ensuremath{\mathsf{VertexConnectivity}}\xspace}

\newcommand{\decide}{\ensuremath{\mathsf{decide}}\xspace}
\newtheorem{definition}{Definition}
\newtheorem{theorem}{Theorem}
\newtheorem{lemma}{Lemma}

\usepackage[capitalise]{cleveref}
\crefname{line}{line}{lines}

\def\BibTeX{{\rm B\kern-.05em{\sc i\kern-.025em b}\kern-.08em
    T\kern-.1667em\lower.7ex\hbox{E}\kern-.125emX}}

\begin{document}

\title{Partition Detection in Byzantine Networks
}


\author{\IEEEauthorblockN{Yérom-David Bromberg\textsuperscript{\rm 1}, Jérémie Decouchant\textsuperscript{\rm 2}, Manon Sourisseau\textsuperscript{\rm 1}, François Taïani\textsuperscript{\rm 1}}
\IEEEauthorblockA{\textsuperscript{\rm 1}\textit{Univ Rennes, Inria, CNRS, IRISA}, France, \textsuperscript{\rm 2}\textit{TU Delft}, The Netherlands\\
david.bromberg@irisa.fr, j.decouchant@tudelft.nl, manon.sourisseau@inria.fr, francois.taiani@inria.fr }
}



\maketitle

\thispagestyle{plain} 
\pagestyle{plain}

\begin{abstract}
Detecting and handling network partitions is a fundamental requirement of distributed systems. Although existing partition detection methods in arbitrary graphs tolerate unreliable networks, they either assume that all nodes are correct or that a limited number of nodes might crash. In particular, Byzantine behaviors are out of the scope of these algorithms despite Byzantine fault tolerance being an active research topic for important problems such as consensus. Moreover, Byzantine-tolerant protocols, such as broadcast or consensus, always rely on the assumption of connected networks. This paper addresses the problem of detecting partition in Byzantine networks (without connectivity  assumption).
We present a novel algorithm, which we call \solname, that safely detects partitioned and possibly partitionable networks and prove its correctness. \solname allows all correct nodes to detect whether a network could suffer from Byzantine nodes. We evaluate \solname's performance and compare it to two existing baselines using up to 100 nodes running real code, \editms{on various realistic topologies.} 
Our results confirm that \solname maintains a 100\% accuracy while the accuracy of the various existing baselines decreases by at least 40\% as soon as one participant is Byzantine. Although \solname's \mschange{network cost} increases with the number of nodes and decreases with the network's diameter, it does not go above around 500KB in the worst cases. 
\end{abstract}

\begin{IEEEkeywords}
Byzantine Fault, Network Partition, Digital Signatures, Distributed Algorithms
\end{IEEEkeywords}


\section{Introduction}
In typical distributed systems, 
messages transit on communication channels 
that form an incomplete network. Two distant nodes only need to be connected through a path of nodes and channels to communicate. In such systems, a partition occurs when the network becomes separated into two (or more) parts that cannot communicate. Partitions can occur for various reasons, such as hardware failures, network congestion, or software issues. They can inherently lead to data loss and data inconsistency, as, for instance, commonly illustrated in the distributed database field and popularized through the P of the CAP theorem~\cite{brewer10}.  

The ability to detect partitions is a key requirement for a wide range of distributed algorithms that (i) assume a connected network, (ii) need to detect specific failures (node crashes), (iii) require information about the topology
(e.g., in MANETs~\cite{ruiz2015survey}), or (iv) auto-configure systems (e.g., in IoT networks~\cite{akkaya2005survey}). 
In such cases, detecting and resolving partitions quickly is essential to maintain network reliability and availability, identify and resolve underlying network issues, improve overall network performance and reduce downtime. 

The partition detection problem has been well studied over the last decade~\cite{ritter_partition_2004, martin_network_2020, bouget_mind_2018, renesse_gossip-style_1998, conan_failure_2008}. However, while some of the existing solutions tolerate crash faults, they break down rapidly when faced with arbitrary, or even malicious behaviors. Such behaviors are often modeled as \emph{Byzantine faults}~\cite{lamport_byzantine_1982}. Byzantine nodes can behave arbitrarily and, therefore, can act in the worst possible way for the correctness of the system. 

Although Byzantine faults have been widely studied in a broad range of contexts~\cite{lamport_byzantine_1982, bracha_asynchronous_1987, dolev_unanimity_1981,bonomi_practical_2021,raynal_fault-tolerant_2018, guerraoui_dynamic_2020,auvolat_byzantine-tolerant_2021, zhang_reaching_2022,augustine2022byzantine}, 
almost all existing works assume a connected network.\footnote{The work of Augustine, Molla, Pandurangan, and Vasudev~\cite{augustine2022byzantine} is the only exception we know of. It is, however, limited to \emph{congested cliques} (a type of fully-connected communication networks).}
In particular, to the best of our knowledge, no work has studied the problem of detecting partitions in \emph{arbitrary} network topologies in the presence of Byzantine nodes.

Detecting a network partition under a Byzantine fault model is however far from trivial. From an operational point of view, correct nodes typically need to decide whether 
their communication with other correct nodes is guaranteed despite Byzantine nodes.
Although simple to state, such a property can be difficult to assess in a distributed setting, in particular if the graph of communication is not immediately known to nodes. This is because the position of Byzantine nodes is not known to correct nodes, and because Byzantine nodes may behave correctly during the partition detection algorithm (thus hiding the danger they represent), and only disrupt communication at some critical later point if they happen to be located at key positions in the network.


To apprehend this difficulty, we introduce the operational notion of \textit{t-Byzantine partitionability} and link it to the vertex-connectivity of the underlying communication graph. We then formally specify the problem of \emph{Byzantine-resilient network partition detection}, and introduce \solname, the first algorithm to solve this problem in a synchronous communication model with signatures, along with its proof of correctness.
Contrary to earlier Byzantine algorithms on arbitrary graphs~\cite{dolev_unanimity_1981,dolev_authenticated_1983}, \solname's does not require any minimal connectivity value and conserves its safety and liveness properties on any graph. \\



{\textbf{Contributions.}} We make the following contributions:

\begin{itemize}
    \item We formally define a new problem, \emph{Byzantine-resilient network partition detection}, and present \solname, a distributed algorithm that solves this problem on any arbitrary graph in the presence of Byzantine nodes.
    \item We formally prove that \solname solves Byzantine-resilient network partition detection .
    \item We have implemented a prototype of \solname in C++ on top of a real network stack on real machines.
    \item We thoroughly evaluate our approach \textit{via} an in-depth experimental evaluation, comparing both network costs and Byzantine reliability to that of two baselines.
    \item We demonstrate that state-of-the-art algorithms lose around 40\% of accuracy in the presence of a single Byzantine node, while our solution remains correct, even in the presence of multiple Byzantine faults. This robustness comes at a reasonable \mschange{network cost}, which does not exceed 500\,KB per node, in the worst case, for a 100-node system. 
\end{itemize}

The rest of this paper is organized as follows.
Section~\ref{sec:system_model} describes our system model. 
Section~\ref{sec:partitions} recalls the definition of a partitioned network, defines $t$-Byzantine partitionable networks, and Byzantine-resilient partition detection algorithms.
Section~\ref{sec:algo} presents \solname, our network partition detection algorithm, and proves its properties. 
Section~\ref{sec:xp_eval} evaluates \solname's performance and compares it to two non-Byzantine-resilient baselines.
Section~\ref{sec:sota} discusses the related work.
Finally, Section~\ref{sec:conclusion} concludes this paper. 

\section{System Model}        
\label{sec:system_model}

We consider a set $\Pi = \{p_1, p_2, ..., p_n\}$ of $n$ processes, each identified by a unique ID. 
All processes know the total number of processes, $\mathbf{card}(\Pi)=n$, and their ID. 
Processes are interconnected by a network represented by a static undirected graph $G = (V, E)$ where $V = \Pi$ and where $E$ contains the communication channels between pairs of nodes. We use $k$ to denote the vertex connectivity of $G$. For simplicity, we assimilate each node of $G$ with the process it hosts, and will use the words `\emph{process}' and `\emph{node}' interchangeably.

Two nodes can directly communicate through a communication channel if and only if they are connected by an edge in $G$. Otherwise, they must rely on other nodes to relay their messages. We assume that communication channels are reliable, and
 that the network is synchronous. In other words, there is a bound $\Delta T$ so that messages sent at a time $T_0$ are always received before time $T_1 = T_0 + \Delta T$. We assume that the local processing time of nodes is negligible compared to communication delays.

Nodes can sign and authenticate messages using an asymmetric digital signature scheme. 
We note $\sigma_i(msg)$ a message $msg$ signed by node $i$. Our protocol leverages chained signatures. For example, $\sigma_j(\sigma_i(msg))$ is a chain signature of message $msg$ that allows a node to extract and verify $\sigma_i(msg)$ and $\sigma_j(\sigma_i(msg))$. 

The \textit{neighborhood} of a node $i$---noted  $\Gamma(i)$---is the set of nodes with which it can directly communicate in $G$, i.e. $\Gamma(i)=\{j\in V\mid (i,j) \in E\}$. When $j\in\Gamma(i)$, we say that $j$ is a \textit{neighbor} of $i$.


We assume that the system contains up to $t$ Byzantine nodes. Byzantine nodes may deviate arbitrarily from their specified protocol, e.g., they may drop, modify, or inject messages at any time. \editms{Following the standard Byzantine fault model~\cite{lamport_byzantine_1982, dolev_unanimity_1981, dolev_authenticated_1983}, Byzantine nodes may not, however, violate network assumptions, such as synchrony (belated messages from Byzantine nodes are just ignored by correct nodes) or reliable channels. In particular, Byzantine nodes cannot prevent two correct neighbors from communicating with each other. 
Byzantine nodes cannot forge signatures, ensuring the authenticity and integrity of messages forwarded by correct nodes.
They cannot spawn new nodes or generate new identities, eliminating Sybil attacks \cite{levine_survey_2006}}.

\editft{Nodes do not know $G$, the underlying network's topology, but know their individual neighborhood $\Gamma(i)$ (either because it was statically configured at set-up, or because it is provided by the underlying network stack, which we assume immune to Byzantine interference).
Each node $p_i$ has further access to a cryptographic \emph{proof of neighborhood} (noted $\proofWord_{p_i, p_j}$) for each of its neighbors $p_j\in\Gamma(i)$. Byzantine nodes cannot forge $\proofWord_{p_i, p_j}$ if it involves at least one correct node. Byzantine nodes may however forge proofs of neighborhood between Byzantine processes.}

\section{Byzantine Partition Detection}
\label{sec:partitions}

This section introduces the notion of $t$-Byzantine partitionable network, and formally define the problem of Byzantine-resilient network partition detection. 
%

\subsection{Network partition}

Informally, a network partition corresponds to the impossibility for pairs of nodes in a network to communicate, even if they rely on intermediary nodes to relay their messages.
More formally, this can be defined as follows: 
\begin{definition}[Network partition] \label{def:partition}
    A communication network $G = (V,E)$ is partitioned if there is a partition $P=\{V_1, \cdots, V_k\}$ of $V$ in $k \ge 2$ subsets such that $\forall (u, v) \in V_i \times V_{j \neq i},\ (u, v) \notin E$. 
\end{definition}

Unfortunately, whether Def.~\ref{def:partition} is satisfied does not characterize the impossibility for pairs of correct nodes to communicate reliably in the presence of Byzantine processes since Byzantine nodes might be able to prevent correct nodes from exchanging messages in a non-partitioned graph. We introduce the notion of $t$-Byzantine partitionability to capture this notion.  

\subsection{t-Byzantine Partitionability}


We say that a graph is \emph{$t$-Byzantine partitionable} if it contains two correct nodes that might be unable to exchange messages if $t$ other nodes are Byzantine. More formally, $t$-Byzantine partitionability is defined as follows. 

\begin{definition}[$t$-Byzantine partitionable graph]\label{def:t:Byzantine:part}
A communication graph $G$ is $t$-Byzantine partitionable if all algorithms executing on $G$ have at least one execution in which at least one pair of correct nodes cannot exchange messages. 
\end{definition}

In practice, the notion of $t$-Byzantine partitionable network does not imply that pairs of correct nodes can never exchange messages. Correct nodes might still be able to communicate if strictly less than $t$ nodes are Byzantine, or if the Byzantine nodes continue to relay messages.
\editms{This notion is directly related to the vertex connectivity of the graph $G$, as stated by the following theorem and its corollary.}

\begin{theorem} \label{def:byzpart}
    A network $G=(V, E)$ is $t$-Byzantine partitionable \editft{iff} there is a set $V_b \subset V$ of $t$ nodes \editft{or less} such that the subgraph induced by $V \setminus V_b$ is partitioned. 
\end{theorem}

\begin{proof}
\newcommand{\algo}{\mathcal{A}}
\editft{ $\bullet$ Let us assume that a network $G=(V,E)$ is $t-$Byzantine partitionable, and consider a simple algorithm $\algo$ in which, in round $1$, every node sends its signed identifier to all its neighbors. Then, starting in round $r\geq 2$, each node retransmits in round $r+1$ all the messages received in round $r$. By definition of $t$-Byzantine partitionability, there exists an execution $E_\algo$ of $\algo$ in which at least one pair of correct nodes $(v_i, v_j)$ cannot exchange messages, in particular $v_j$ never receives $v_i$'s signed identifier. Let us note $V_b^{\algo}$ the set of Byzantine nodes in $E_\algo$. Because all nodes in $V\setminus V_b^{\algo}$ are correct, they execute $\algo$ faithfully. The fact that $v_j$ never receives $v_i$'s signed identifier implies there is no path between $v_i$ and $v_j$ in $V\setminus V_b^{\algo}$, and therefore that $V \setminus V_b^{\algo}$ is partitioned.

$\bullet$ Turning to the reverse implication, consider a network $G=(V,E)$ such that there is a set $V_b \subset V$ of $t$ nodes such that the subgraph $G_b$ induced by $V \setminus V_b$ is partitioned. Consider an algorithm $\algo$ executing on $G$. Consider an execution $E_\algo$ of $\algo$ in which all nodes in $V_b$ are Byzantine (which is possible as $|V_b|=t$) and all nodes in $V\setminus V_b$ are correct. Assume all Byzantine nodes drop all messages they receive. As a result, in $E_\algo$, the messages of $\algo$ sent by correct nodes may only travel on the edges of $G_b$, the graph induced by $V \setminus V_b$. As $G_b$ is partitioned, there exist two correct nodes in $V\setminus V_b$ that cannot exchange messages in $E_\algo$.}
\end{proof}

Since \editft{in a given execution} the identity of the $t$ possible Byzantine processes is unknown, using Theorem~\ref{def:byzpart} to test whether a graph $G$ is t-Byzantine partitionable requires that every possible set $V_b$ of $t$ processes be tested.
\editft{Such a test directly translates into a condition on the \emph{vertex-connectivity} of the graph $G$, defined as the size of the smallest vertex subset of $G$ that partitions $G$.}

This translates into the following condition: 

\begin{corollary}\label{prop:connectivity}
    A network $G=(V,E)$ is $t$-Byzantine partitionable iff. its vertex connectivity is lower than or equal to $t$.  
\end{corollary}

\begin{proof}
    \editft{This directly follows from the \cref{def:byzpart} and the definition of the vertex connectivity of a graph.}%
\end{proof}

\editft{In a graph with connectivity $k$ larger than $t$, the subgraph of correct nodes remains connected no matter how the $t$ Byzantine nodes are placed. As a result, correct nodes can continue to exchange messages (possibly indirectly) independently of the Byzantine nodes' behavior. For instance, the graph in \cref{subfig:2:connected} is 2-connected. With $t=1$, a single Byzantine node cannot prevent the remaining correct nodes from communicating with each other, whichever its position in the graph.

By contrast, $k \le t$ does not necessarily imply that correct nodes cannot communicate. The disruption that Byzantine nodes may cause will depend in this case on their position in the graph. However, $k \le t$ implies that at least one such placement exists in which the $t$ Byzantine nodes can prevent correct nodes from communicating with each other. 
In the star graph shown in \cref{subfig:1:connected}, if $t=1$, a Byzantine node will prevent correct nodes from communicating with each other only if it is placed in the center position.
}

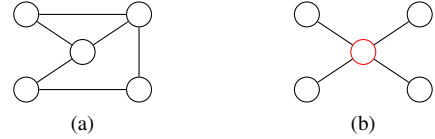
\begin{figure}[tb]
    \hfill
    \subfloat[][]{\label{subfig:2:connected}          
        \begin{tikzpicture}
          \node[shape=circle, draw=black] (F) at (-1,0) {};
          \node[shape=circle, draw=black] (G) at (0.5,0) {};
          \node[shape=circle, draw=black] (H) at (-0.25,-0.5) {};
          \node[shape=circle, draw=black] (I) at (0.5,-1) {};
          \node[shape=circle, draw=black] (J) at (-1,-1) {};
          \path[] (F) edge node[left] {} (H) ;
          \path[] (G) edge node[left] {} (H) ;
          \path[] (J) edge node[left] {} (H) ;
          \path[] (F) edge node[left] {} (G) ;
          \path[] (G) edge node[left] {} (I) ;
          \path[] (I) edge node[left] {} (J) ;
    \end{tikzpicture}
    }\hfill{}
    \subfloat[][]{\label{subfig:1:connected}
    \begin{tikzpicture}
         \node[shape=circle, draw=black] (A) at (-5,0) {};
          \node[shape=circle, draw=black] (B) at (-3.5,0) {};
          \node[shape=circle, draw=red] (C) at (-4.25,-0.5) {};
          \node[shape=circle, draw=black] (D) at (-3.5,-1) {};
        \node[shape=circle, draw=black] (E) at (-5,-1) {};

           \path[] (A) edge node[left] {} (C) ;
          \path[] (B) edge node[left] {} (C) ;
          \path[] (D) edge node[left] {} (C) ;
          \path[] (E) edge node[left] {} (C) ;
    \end{tikzpicture}
    }
    \hfill
    \hbox{}

    \caption{%
     (a): A graph that is not 1-Byzantine partitionable. No matter the placement of a Byzantine node, the subgraph of correct nodes remains connected.
     (b): A 1-Byzantine-partitionable graph. If the red node is Byzantine, then the subgraph of the correct nodes is partitioned. 
    }%
    \label{fig:example_partitiable}
\end{figure}

\subsection{Byzantine network partition detection}
Ideally, a partition detection algorithm should ensure that all correct nodes reach the same correct conclusion: either that they will remain able to communicate with one another, independently of the behavior of Byzantine nodes, or that they will not.
Unfortunately, the boundary between these two situations is not as clear-cut as it seems.
This is because $t$-Byzantine partitionability (\cref{def:t:Byzantine:part}) captures the existence of a worst case where at least one placement of Byzantine nodes can disrupt the communication between correct nodes. This worst case is not guaranteed and only occurs if Byzantine nodes form a vertex cut of the graph $G$. If they do not, even a $t$-Byzantine partitionable graph $G$ will allow correct nodes to communicate with one another independently of Byzantine behaviors. 

Because correct nodes do not know which nodes are Byzantine, and Byzantine nodes might pretend to be correct (at least while the partition detection algorithm executes), correct nodes cannot distinguish a favorable placement from an unfavorable one. In a Byzantine context, a partition detection algorithm must, therefore, foresee two outcomes:
\begin{itemize}
    \item \notPartitionable. No placement of Byzantine nodes can disconnect correct nodes.
    \item \partitionable. Byzantine nodes might be able to disconnect correct nodes (but this is not certain).
\end{itemize}

$t$-Byzantine partitionability (and vertex connectivity) might appear as a natural metric to characterize these two cases.
Unfortunately, correct nodes initially do not know $G$, and must instead collaborate to gather information regarding $G$. Byzantine nodes can disrupt this process, and distort 
 how correct nodes perceive $G$ (by not communicating some edges, for example). 
As a result, if 
a correct process observes a connectivity of $t$, it cannot
distinguish between the following two situations: 
(i) the network connectivity is indeed equal to $t$ and the network is $t$-Byzantine partitionable; and (ii) the network connectivity is in fact higher than $t$ but some Byzantine nodes omitted to send some messages, thus corrupting this correct node's perception. In this case, the correct node can only conclude that the network \emph{might} be $t$-Byzantine partitionable. 
One direct consequence of this scenario is that 
correct nodes might conclude to a partitionable graph, where the real connectivity is great enough to avoid partition. 


\subsection{Formal Specification}
\editft{In the light of the previous discussion, we define a network partition detection algorithm as a distributed algorithm that provides each node with a call-back $\decide()$ that returns one of \editft{two} possible values: \notPartitionable, or \partitionable. 
Correct processes execute $\decide()$ once\footnote{The specification and the algorithm we present assume a static graph and are therefore \emph{one-shot}. In practical cases, the connectivity graph might, however, evolve over time. In such cases, we assume that the graph remains static long enough for the algorithm to execute.}.}
\editft{The algorithm is further parameterized by a known threshold $k_0>t$ that characterizes its tendency to make conservative decisions.}
\begin{definition}[\editft{$t$-Byzantine-resilient, $k_0$-sensitive network partition detection}] \label{def:properties}
We say that a network partition detection algorithm $\mathcal{A}$ is $t$-Byzantine-resilient and \editft{$k_0$-sensitive} if it satisfies the following properties \editft{for any communication graph $G$ and placement of Byzantine nodes $V_b$ (with $|V_b|\leq t$)}:
\editms{\begin{itemize}
    \item \editms{\textbf{Termination.} All correct nodes decide within a bounded amount of time.}
    \item \textbf{Agreement.} \editft{All correct nodes decide the same value.}
    \item \editft{\textbf{Safety.} If Byzantine nodes can effectively prevent communication between correct nodes (formally, if \editft{$V_b$ is a vertex cut} of $G$) then no correct process ever decides \notPartitionable.}
    \item \editms{\editms{\textbf{$k_0$-Sensitivity.} \editft{If $G$'s connectivity is higher or equal to $k_0$}, then all correct nodes decide \notPartitionable.}}
\end{itemize}}
\end{definition}

The algorithm we present in the next section, \solname, fulfills the properties of \cref{def:properties} with $k_0=2t$. It also provides an additional Boolean output, \confirmed, that indicates that a correct node \editms{has detected an actual partition, i.e., 
that Byzantine nodes can effectively filter/cut off communications between correct nodes.} 
Formally, this additional output fulfills the following property:
\begin{itemize}
    \item \textbf{Validity.} If a correct node computes \confirmed $= \true$, then $V_b$ is a vertex cut of the graph $G$. 
\end{itemize}

\section{\editms{\solname: Neighbors Exploring Connections Toward Adversary resilience}}
\label{sec:algo}

\solname solves the $t$-Byzantine-resilient, \editft{$2t$-sensitive} network partition detection problem that we have just defined. \solname does not require nodes to know the underlying network topology, except for the identity of their immediate neighbors (in the form of neighborhood proofs).
\editft{To the best of our knowledge, it is the first algorithm to solve the problem of network partition detection in arbitrary graphs in a Byzantine context.}

\subsection{\editft{Overview}}


\textbf{Inputs/Output.} \solname executed at a process $p_i$ takes 4 parameters as input: (i) the number of nodes in the system ($n$); (ii) the maximum number of Byzantine nodes ($t$); (iii) the neighborhood $\Gamma(i)$ of the local node $i$; (iv) a proof of neighborhood for each of its neighbors ($\proofWord_{i,j}$ for $j \in \Gamma(i)$).

\editft{A \emph{proof of neighborhood} $\proofWord_{p_i, p_j}$ is 
a cryptographic object used by $p_i$ to declare an edge with process $p_j$ that cannot be forged as soon as either $p_i$ or $p_j$ is correct.}

\solname's output is one of two possible decisions: \notPartitionable and \partitionable. \editms{It provides moreover an additional indicative Boolean output, \confirmed.} The \notPartitionable value corresponds to a situation where the graph is not partitioned and cannot be partitioned by the $t$ Byzantine nodes. \partitionable corresponds to the case where the node evaluates that the graph is connected but that its connectivity appears lower than $t$, i.e., the graph is $t$-Byzantine partitionable. The \confirmed boolean identifies a partitioned network \editms{where some nodes are unreachable.} \\

\textbf{Intuition.} The key idea of \solname is to make nodes estimate the vertex-connectivity of the network to detect whether the network is partitioned or could be partitioned by $t$ Byzantine processes not relaying messages.  To do so, nodes disseminate their edges in signed messages and evaluate the size and connectivity of the graph they obtained after a limited number of rounds. \editms{Focusing on communication edges instead of trying directly to communicate with every node is one way to avoid correct nodes to conclude to a partition situation, when faulty nodes act as crashed ones, for example, even if the communication network is actually connected. It, however, strongly relies on the hypothesis that each node knows its neighborhood (and can prove it). The use of signature chains, with a length corresponding to the number of rounds, limits the impact of Byzantine nodes but also ensures that all correct nodes will reach the same conclusion. Sec~\ref{sec:specification} gives the intuition on how \solname satisfies the properties of Def~\ref{def:properties}. }


\subsection{Pseudocode} 

Alg.~\ref{algo1} details \solname's pseudocode, which we discuss in the following. \\

\textbf{Initialization.} Initially, each node keeps in memory an adjacency matrix that will contain all the edges it discovers during the algorithm's execution. 
This adjacency matrix is initialized as follows (ll.~\ref{algo:init_phase}-\ref{algo:end_init_phase}): for each neighbor, the matrix contains locally-generated proofs of the neighborhood. For example, if $i$ and $j$ are neighbors, then $i$ will initialize $G_i$ such that $G_i[i,j] = \proofWord_{i,j}$. Each node also initializes empty buffers $(to\_be\_sent\_i)_i$, that are necessary for the synchronous communication phase. \\

\textbf{Edge Propagation phase.} This phase comes right after the initializing phase (ll.~\ref{algo:edge_prop}-\ref{algo:end_edge_prop}). Nodes communicate synchronously during $n-1$ communication rounds. This phase allows nodes to share their knowledge of the graph, i.e., transmit the edges they know to their neighbors. In particular, nodes send their neighborhood (with signed proofs) during the first round (ll.~\ref{algo:first_round}-\ref{algo:end_first_round}). When a node receives a signed message, it verifies it and relays it with its signature, which generates signature chains, to its neighborhood at the beginning of the next round (ll.~\ref{algo:next_round}-\ref{algo:end_next_round}). Upon reception of a signature chain (ll.~\ref{algo:rcv_msg}-\ref{algo:end_edge_prop}), a process checks its correctness, but also its length (l.~\ref{algo:check_length}). A correct execution implies that the length of the signatures chain is always equivalent to the current round number. This verification prevents Byzantine nodes from transmitting late messages. 
To avoid over-flooding, nodes make sure not to resend an edge they have already sent (l.~\ref{algo:check_length}). \\

\textbf{Decision phase.} After $n-1$ rounds, all correct nodes have the same adjacency matrix and are able to compute if all nodes are reachable and if the vertex-connectivity is greater than $t$, to determine if the graph is partitioned or $t$-Byzantine partitionable (ll.~\ref{algo:decision}-\ref{algo:end_decision}). Note that to accurately evaluate the topology, the number of rounds to use $R$ should be larger than the graph diameter. Since we assume that nodes do not know the topology, $n-1$ rounds are the lowest $R$ value one can use (the worst case being the chain topology). Choosing a different value for $R$ does not change the message complexity of \solname since no node will learn a new edge after the round that corresponds to the graph diameter, and all correct nodes will stay silent in the following rounds. \\

  \newcommand{\announcePhase}[1]{%
    \smallskip%
  
    {\SetKwComment{Comment}{}{}\Comment{\fbox{#1}}}
    \vspace{0.25em}%
  }

\begin{algorithm}[tb]
\label{algo1}
\small
\Indm
\Input{$n, t, \Gamma(i), (\proofWord_{i,j})_{j\in \Gamma(i)}$ \Comment*{for node $i$} }
\Output{\decide $\left\{\begin{array}{@{}l@{}}\text{\notPartitionable or} \\\text{\partitionable}\end{array}\right.$} 
\Indp
\announcePhase{Initialising $G_i$} 
$G_i = matrix(n, n)$ \Comment*[f]{empty matrix} \; \label{algo:init_phase}  
\ForEach{$\neighbor \in \Gamma(i)$}{$G_i[i,\neighbor] \leftarrow \proofWord_{i,\neighbor}$} 
$to\_be\_sent_1$ = \o \label{algo:end_init_phase} \;
\announcePhase{Edge Propagation}
\Round{$R \in [1,n-1]$ \label{algo:edge_prop}}{
\eIf{$R = 1$\label{algo:first_round}}{
\ForEach{$\neighbor \in \Gamma(i)$}{
\send ${\big \{}\sigma_i(\proofWord_{i,j}){\big \}}_{j \in \Gamma_i }$ \sendto $\neighbor$}\label{algo:end_first_round}
}{
\ForEach{$(msg, k) \in to\_be\_sent_{R}$ \label{algo:next_round}}{
\ForEach{$\neighbor \in \Gamma(i) \setminus k$}{\send $\sigma_i(msg)$ \sendto $\neighbor$ \label{algo:end_next_round}
}}}

\WhenReceived{$msg = \sigma_k(\sigma_x(\sigma_y(...., \sigma_u(\proofWord_{u,v}))))$ {\bf from} $k$ \label{algo:rcv_msg}}{
    \Comment{Invalid messages are ignored}
    \If{$\lengthSign(msg) = R$ and $G_i[u][v] \neq \textsf{nil}$ \label{algo:check_length}}{
    
   \ensuremath{\mathsf{add}}\xspace $(msg, p_k)$ \ensuremath{\mathsf{to}}\xspace $to\_be\_sent_{R+1} $ 
    $G_i[u][v] = \proofWord_{u,v}$} 
}}\label{algo:end_edge_prop}
\announcePhase{Decision}
$r = \reachable(G_i)$ \label{algo:decision} \;
$k = \connectivity(G_i)$ \;
\eIf{$k > t$ and $r = n $}
{ $\confirmed = \false$ \; decide(\notPartitionable) }
{ 
\lIf{$r = n$}
{ $\confirmed = \false$}
\lElse{ $\confirmed = \true$} \label{algo:end_decision}
decide(\partitionable) \;
}
\caption{\solname{}'s pseudocode}
\end{algorithm}


\textbf{Impact of Byzantine deviations}. \solname limits Byzantine nodes' ability to lie about their neighborhood or create knowledge disparities among correct nodes. However, it cannot compel Byzantine processes to share their own neighborhood correctly. In particular, edges that connect two Byzantine nodes might never be discovered, which might decrease the graph's vertex connectivity below $t$. In this case, correct nodes will decide that the graph is \partitionable, while it is not in reality. However, pairs of Byzantine nodes that declare fictitious edges connecting them are not an issue because these edges will never increase the vertex-connectivity above $t$ if the subgraph of correct nodes is partitioned. In this case, the correct nodes will decide that the graph is  \partitionable. 

\subsection{Agreement\editms{, Safety, and Sensitivity} Properties} \label{sec:specification}

We now \editms{give the intuition on} how \solname enforces the properties of Def.~\ref{def:properties}, which are more formally proven in Sec.~\ref{sec:proofs}.

The values that correct processes decide depend on the graph connectivity and on the connectivity of the correct process subgraph. We identify three main cases: 

$\bullet$ If $2t \le k$ (case 1), then all correct nodes decide \notPartitionable.

$\bullet$ The case $k \le t$ is divided in two subcases. First, if $0 < k \le t$ and the subgraph of correct nodes is connected (case 2.1), then all correct nodes decide \editms{\partitionable (with $\confirmed =\true$ if they detect a real communication issue, \confirmed $= \false$ otherwise)}, or \notPartitionable. \editms{Note that the \notPartitionable value can be decided even if $k \le $t, but is not an issue, as discussed previously about the impact of Byzantine deviations in Sec~\ref{sec:algo}. }
Second, if $0 < k \le t$ and the correct nodes subgraph is not connected (case 2.2), then all correct nodes decide \partitionable with \confirmed $=\true$ if they detect a real communication issue, \confirmed $= \false$ otherwise). In this case, correct nodes do not necessarily reach the same conclusion about \confirmed boolean, but they all decide that the network is \partitionable. For those two cases, note that if $k = 0$, then all correct nodes decide \partitionable (Byzantine nodes cannot increase the connectivity more than $t$). If we consider that Byzantine nodes can exchange messages through another network, correct nodes, however, do not necessarily compute $\confirmed = \true$ 

$\bullet$ If $t < k < 2t$ (case 3), then either all correct nodes decide  \partitionable{} \editms{ (Byzantine nodes might not share some edges, that might decrease the perceived connectivity below $t$)} or they all decide \notPartitionable. However, correct nodes compute \confirmed boolean to $\false$


\subsection{Proof of correctness} 
\label{sec:proofs}

%

\editms{We divide our proof in four lemmas, that refer to the three cases of the previous section (Sec.~\ref{sec:specification})}. \mschange{We note $G_{i,r}$ the state of the adjacency matrix of node $i$ at the end of round $r$ ($r=0$ corresponds to the initial state of $G_i$), $B$ the set of Byzantine nodes, and $C$ the set of correct nodes.}
\begin{lemma} \label{lemma:spec1} \editms{(case 1)}
    If the graph is (at least) $2t$-connected, then all correct nodes decide \notPartitionable.
\end{lemma}

\begin{proof}
    We suppose that the graph is $2t$-connected.
    First, we show that every correct node computes $r = n $.
    As we suppose that the graph is $2t$-connected, by definition, the graph is therefore also $t+1$-connected. Due to Menger's theorem~\cite{menger1927allgemeinen}, there exist at least $t+1$ vertex independent paths between every pair of nodes. Because there are at most $t$ Byzantine nodes, it implies that every Byzantine node is the neighbor of a correct node. With our notations, this property can be written as: 
    \begin{equation}\label{eq:forall:b:exists:i:i:knowns:of:b}
       \forall b \in B,\ \exists i \in C: G_{i,0}[i, b] = \sigma_i(\proofWord_{i,b}) 
    \end{equation}

    We now show that for two correct nodes, $c_1$ and $c_2$, at the end of $n-1$ rounds, we have:
    \begin{equation} \label{eq:forall:c:c:will:know}
        \forall c_1, c_2 \in C,\  G_{c_1,1} \subseteq G_{c_2,n-1} 
    \end{equation}

    Using Menger's theorem~\cite{menger1927allgemeinen}, because the graph is at least $t+1$-connected, there is a path $(p_1, ..., p_m)$ of correct nodes in $G$ such that $c_1 = p_1$, $c_2 = p_m$ and $ m \leq n$. 
    We can trivially show by induction on $k$ ($k \in [1,m-1]$), that $ G_{p_1,0} \subseteq G_{p_{1+k}, k}$, because the neighborhood of node $c_1$ propagates from node to node in each round (assuming synchronous communication) following the correct node path.
    For $k = m-1$, we have thus $G_{c_1,0} \subseteq G_{c_2,m-1}$.

    A corollary of Equations~\ref{eq:forall:b:exists:i:i:knowns:of:b} and~\ref{eq:forall:c:c:will:know} is:
    \begin{equation} \label{eq:nodes:know}
        \forall b {\in} B\cup C,\ \forall i {\in} C,\ \exists k {\in} C: G_{i,n-1}[k, b] = \sigma_k(\proofWord_{k,b}) 
    \end{equation}
    This corollary means that each correct node $i$ sees every node $b$ as reachable, and therefore computes $r = n$.

    Second, we show that every correct node computes $k > t$. 

    A direct corollary of Equation~\ref{eq:nodes:know} is that each node sees every edge between two correct nodes (if $b \in C$) and every edge between correct nodes and Byzantine nodes (if $b \in B$). Because we suppose that the graph is $2t$-connected, \editms{even by removing the edges between Byzantine nodes, which can drop the connectivity by at most $t-1$}, 
    every node computes (at least) a $t+1$-connectivity. 
    Then every node decides value \notPartitionable. 

\end{proof}

\begin{lemma} \label{lemma:spec2} \editms{(case 2.1 and 3)}
    If the subgraph of correct nodes is connected, then all correct nodes decide the same decision.
\end{lemma}

\begin{proof}
    Let us prove that, at the end of round $n-1$, every correct node computes the same graph topology, noted $G_f$:
    \begin{equation} \label{eq:same:graph}
        \forall i,j \in C,\ G_{i,n-1} = G_{j,n-1} = G_f
    \end{equation}

    Because the subgraph of correct nodes is connected, we have:
    \begin{equation} \label{eq:voisins}
        \begin{multlined}
        \forall i, j \in C,\ \exists k \in [1, n-1],\ \forall u \in \Gamma(i),\\  G_{j,k}[i,u] = \proofWord_{i,u}
        \end{multlined}
    \end{equation}
    This comes from the fact that each node propagates its neighborhood to every other reachable node. Since the subgraph of correct nodes is connected, every correct node will receive this neighborhood during the $n-1$ rounds.

Each message sent by a correct node will thus be received by every other correct node. Thus, we have:
    \begin{equation}
        \forall i,j \in C,\ G_{i,0} \subseteq G_{j,n-1}
    \end{equation}

    We now show that $\forall i,j \in C,\ \forall u \in B \cup C,\ \forall b \in B,\   \forall k \in [1, n-1], $
    \begin{equation}
        \begin{multlined}
         G_{i,k}[b,u] = \proofWord_{b,u} \Rightarrow 
         \left\{\begin{array}{l}
              \exists k' \in [1, n-1],\\
              \hspace{1em}G_{j,k'}[b,u] = \proofWord_{b,u}
         \end{array}\right.
         \end{multlined}
    \end{equation}
    Intuitively, if a correct node receives (and accepts) a message from a Byzantine node during the $n-1$ rounds, then every correct node will receive this same message.

    We use the same argument as Dolev and Strong~\cite{dolev_authenticated_1983}: 
    if a message  $msg {=} \sigma_r(\sigma_{r-1}(\sigma_{r-2}(...\sigma_{1}(proof_{b,u}))))$ is received by a correct process for the first time in round $r\leq n-1$ (meaning that no other correct process has received $msg$ in a round $r'\in[1,r-1]$), then all the processes that have signed $msg$ in the rounds $[1,r]$ must be faulty, and $r\leq t-1$. The correct node that is the first to receive the message $msg$  will propagate a signed message (i.e., will add its signature to the chain and forward it to its neighbors) containing $msg$ in round $r+1$,  and all correct processes will receive $msg$ at the latest by round $r+1+d-1\leq t+d-1$ where $d$ is the diameter of the graph of correct nodes ($d-1$ rounds is thus sufficient for a message to travel among the graph of correct nodes). As $d<n-t$ ($n-t$ is the number of correct nodes), all correct nodes receive a signatures chain containing $msg$ by round $n-1$. 

    One can note that this earlier argument makes the assumption that the system contains exactly $t$ Byzantine nodes. If the number of effective Byzantine nodes is lower than $t$, the same reasoning holds.
    
    So finally, 
    \begin{equation}
        \forall i,j \in C,\ G_{i,n-1} = G_{j,n-1} = G_f
    \end{equation}
    
    At the end of the $n-1$ rounds, all correct nodes will have the same view of the graph topology ($ \forall i \in C, G_{i,n-1} = G_f$). Thus, every correct node will compute the same $r$ and the same $k$, leading them to reach the same decision.
\end{proof}

\begin{lemma} \label{lemma:spec3} \editms{(case 2.2)}
    If the subgraph of correct nodes is disconnected, then all correct nodes will decide \partitionable. 
\end{lemma}
\begin{proof}
    Let us assume that the subgraph of correct nodes is disconnected and that a correct node decides \notPartitionable. It means that this correct node has computed a vertex connectivity of at least $t+1$. Thus, it means that it received signatures from every node through at least $t+1$ vertex-disjoint paths. Thus, there is a path between this node and every other correct node that is free of Byzantine nodes, which means that the subgraph of correct nodes is connected and contradicts our assumption.
\end{proof}


\begin{theorem}
    Algorithm~\ref{algo1} is a $t$-Byzantine-resilient network partition detection algorithm (defined in \editms{Def.~\ref{def:properties}}).
\end{theorem}
\begin{proof}
    \editms{Alg.~\ref{algo1} satisfies the termination property due to the network synchrony hypothesis. Agreement is ensured by Lemma~\ref{lemma:spec2} and Lemma~\ref{lemma:spec3}. Safety is ensured by Lemma~\ref{lemma:spec3}. Finally, \solname is $2t$-sensitive, according to Lemma~\ref{lemma:spec1}. Regarding the validity property, if a node computes $\confirmed = \true$, it has in particular, computed $r \neq n$ (Alg.~\ref{algo1} (ll.~\ref{algo:decision}-\ref{algo:end_decision})).  There are thus two cases: first, if the subgraph of correct nodes is disconnected, then $V_b$ is directly a vertex cut of $G$. Second, if the subgraph of correct nodes is connected, then according to Lemma~\ref{lemma:spec2}, every correct node computes $r \neq n$. It thus means that, according to Eq~\ref{eq:voisins}, there exists $b_0 \in B$, such that for every correct node $c \in C$, $b_0$ is not in $\Gamma(c)$, i.e., $V_b$ is a vertex cut of $G$. }
\end{proof}

\subsection{Communication Complexity.} \label{sec:complexity}
\editms{According to the pseudo-code of Alg.~\ref{algo1}, we can compute the message complexity as follows.
Each node has to forward a unique message for each network edge, one for each of its neighbors. In the worst case (fully connected graph), the complexity is thus in $O(n^4)$. A key aspect is that the complexity highly depends on the network topology. The more edges the graph has, the higher the global communication cost is. The communication cost can also be very disparate through nodes since the complexity for each node depends on the size of its neighborhood. One can also note that the \mschange{\mesure cost} through the algorithm might vary a lot: the lower the graph's diameter is, the sooner the nodes will have exchanged their neighborhood information. It might thus happen that in the last rounds of the algorithm, every correct node stays silent because they have already discovered all the edges of the graph. }

\section{Evaluation}
\label{sec:xp_eval}

We compare \solname experimentally to two baselines on different families of network topologies in terms of \mschange{network cost} and resilience to Byzantine behaviors.

\subsection{Baselines}

To the best of our knowledge, there is no existing Byzantine tolerant partition detection algorithm in the literature \editms{that is compliant with any topologies}. We, however, compare \solname to MindTheGap (MtG)~\cite{bouget_mind_2018}, an efficient network partition detection solution. Processes in MtG flood a list of reachable nodes to each other. Nodes keep in memory a list of reachable nodes (that only contains themselves initially), and send regularly this list to their neighbors, during a fixed period of time (an \textit{epoch}). When receiving a list of neighbors, nodes can actualize their own list of reachable nodes. MtG has a low network consumption because it uses Bloom filters to represent a list of process IDs. 

MtG mechanisms are easily corruptible by a single Byzantine process (as explained in Sec.~\ref{sec:xp_byz}). We thus decided to also consider a strengthened version of MtG as a second baseline, where MtG's Bloom filters are replaced by a list of signed process IDs. 
To minimize the increased \mschange{network cost} associated to this modification, we made sure that nodes only send a given signed ID once to their neighbors per epoch. In the following, we call MtGv2 this variant of MtG.

\subsection{Experimental setup} 
\label{sec:xp_and_use_cases}

We implemented all protocols in \texttt{C++} and used the \texttt{salticidae} networking library\footnote{https://github.com/Determinant/salticidae}. We use the ECDSA signature scheme~\cite{johnson_1998_elliptic}. \mschange{We ran our experiments on a Dell PowerEdge R640 server (2x Intel(R) Xeon(R) Gold 6136 CPU @ 3.00GHz, 12 cores/processor, hyper-threading, 188GB of RAM).}
Our deployment code uses one \texttt{Docker} container per process~\cite{merkel_2014_docker}. 

\editms{To study the \mschange{network cost} of \solname, we consider two families of graph topology:
\begin{itemize}
    \item (i) Realistic connectivity-dependent topologies, such as those considered by Bonomi, Farina, and Tixeuil \cite{bonomi2019multi}, since connectivity is a key aspect of our approach.  
    \item (ii) Random graphs that model a simple drone network. 
\end{itemize}

\begin{figure}[tb]
    \centering
    \includegraphics[width = 0.87\columnwidth]{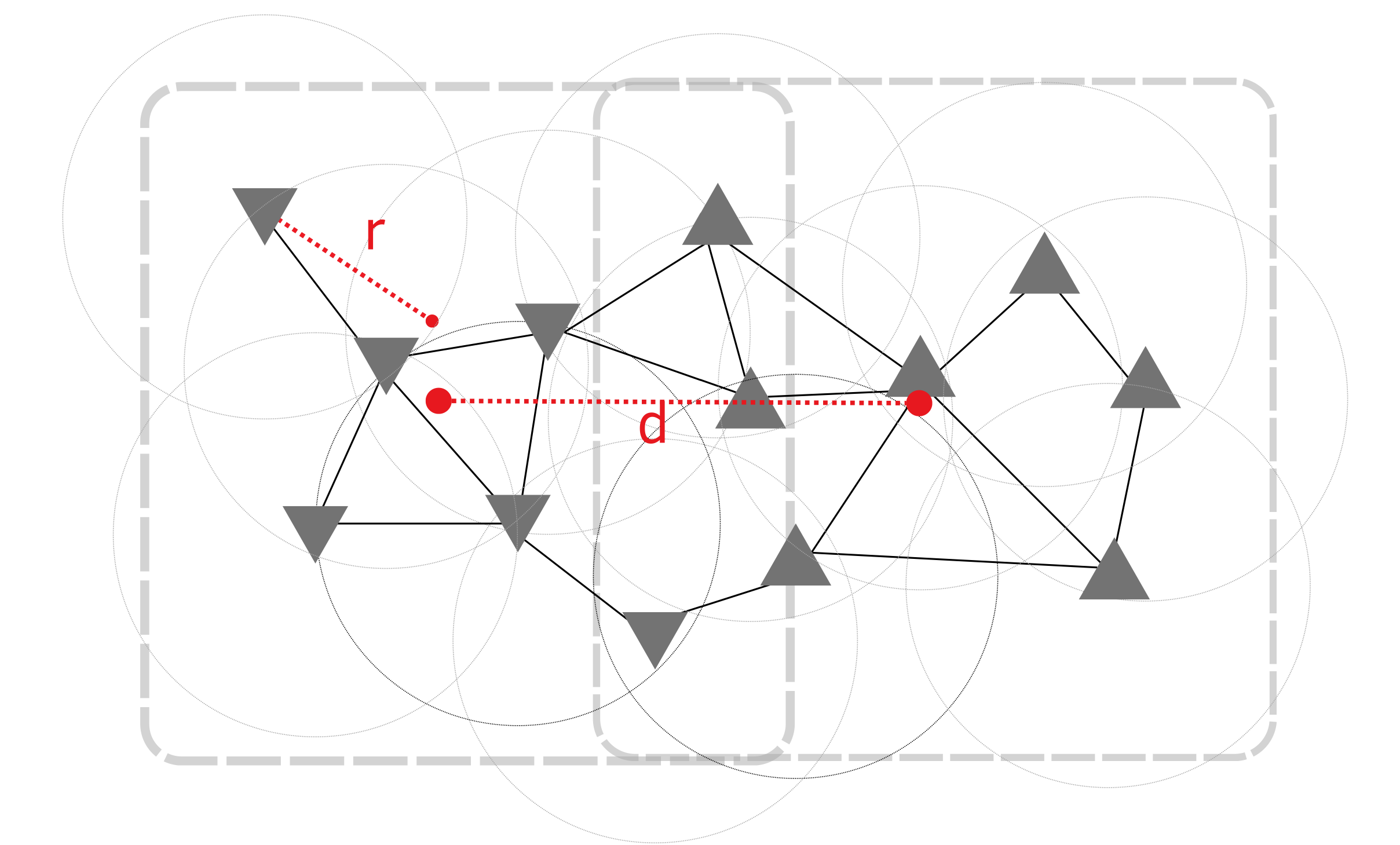}
    \caption{An example of our drone scenario with random graph. Two scatters of points are generated. $d$ is the distance between the barycenters of the scatters, and $r$ is the communication scope.}
    \label{fig:cloud_drone}
\end{figure}

Bonomi, Farina and Tixeuil \cite{bonomi2019multi} highlight several class of graphs, such that:
\begin{itemize}
    \item $k-$regular $k-$connected graphs \cite{steger_wormald_1999}. Regular graphs ensure that the graph's connectivity is exactly $k$ (with the minimum number of edges) and that each node has exactly $k$ neighbors. 
    \item $k-$pasted-tree and $k-$diamond graphs \cite{BALDONI20092110}. Those topologies are \textit{Logarithmic Harary Graph}, built to have interesting properties for fault-tolerance and suit message flooding communication protocols. 
    \item Generalized and Multipartite Wheel graphs \cite{bonomi2019multi}. Those topologies are the worst-case scenarios while considering Byzantine faults: Byzantine nodes might compose a clique while it might have only one (generalized wheel) or few (multipartite wheel) path(s) that link all correct nodes (and thus increase latency).   
\end{itemize}
}
For the drone scenario, we create random graphs by generating random nodes in a 2D space, and a scope parameter decides edges: if two nodes are close enough (i.e., their distance is lower than $radius$), then we add an edge between them.  Those nodes are randomly generated around two barycenters. In the following, we vary the distance between barycenters (noted $d$), the scope parameter ($radius$), and the number of nodes ($n$). \editms{Fig~\ref{fig:cloud_drone} is an example of our scenario, which aims to model a drone network, where two drone scatters are moving away or approaching in space.} For each situation, we run \editms{50} times the experimentation and report average results. \editms{Error intervals correspond to a confidence interval of 95\% we obtained during our experimentation. }


We finally selected a number of nodes to act as Byzantine nodes to study the Byzantine resilience of the tested algorithms. The tested Byzantine behaviors are discussed in Sec.~\ref{sec:xp_byz}.

\subsection{\ftchange{Network cost}}

\textbf{Testing impact of connectivity.}
\editms{We discuss in this section of the performance of \solname on the topologies highlighted by Bonomi and al.~\cite{bonomi2019multi}, while varying the connectivity parameter. }
Fig.~\ref{fig:Regular_fig} shows the \mschange{data sent} per node depending on the total number of nodes for several vertex-connectivity parameters $k$ \editms{for $k-$ regular $k$-connected graphs.} In the worst cases, ($n = 100$ and $k = 34$), the \mschange{\mesure} per node is around $500$\,KB, which is correct for nowadays technologies. Such a figure allows us to know the \mschange{network cost}, depending on the wanted robustness for several sizes of systems, up to 100 nodes. 
\newcommand{\expeFigScaling}{0.5}

\begin{figure}[tb]
    \centering
    \includegraphics[scale=\expeFigScaling]{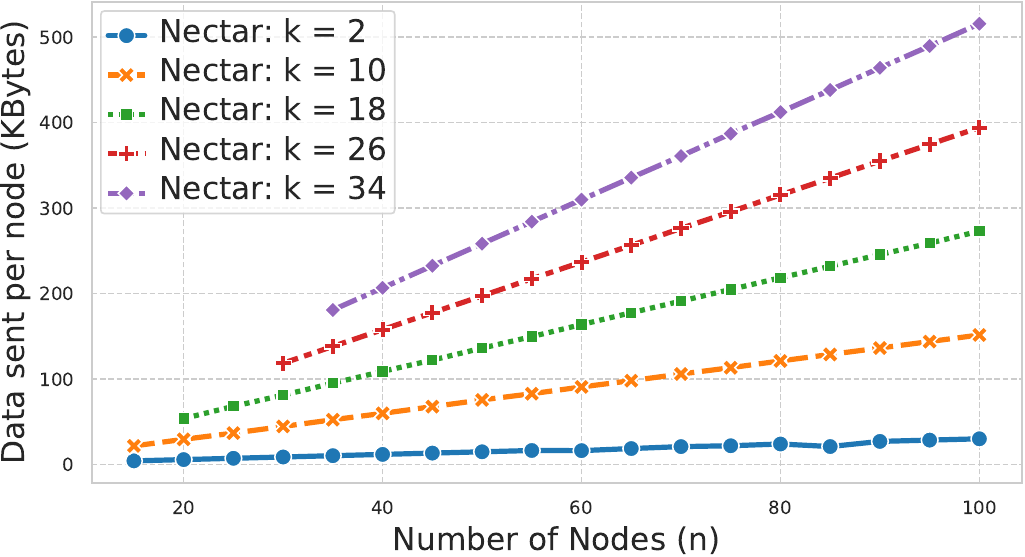}
    \caption{\ftchange{Data sent} per node (in KB), depending on the number of nodes ($n$), for different vertex-connectivity ($k$) in $k-$regular graphs, for \solname }
    \label{fig:Regular_fig}
\end{figure}

\editms{On the others topologies \solname exhibits similar behaviors but seems less costly. In settings identical to those of Fig.~\ref{fig:Regular_fig}, \solname is around 2 times less costly on $k-$diamond graphs and $k-$pasted graphs, and around 2.5 times less costly on multipartite wheel graphs and generalized wheel graphs.} \\

\textbf{Drone based scenario} 
Fig.~\ref{fig:random_BFT} shows the \mschange{network cost}  per node in the drone scenario, depending on the distance ($d$) between the two barycenters, with $n = 20$ nodes. For $d = 0$ and $radius = 2.4$, the graph corresponds to a fully connected graph, and the \mschange{\mesure} is around $50$\,KB. 
A distance $d = 6$ corresponds to a partitioned graph of two parts (it can be more for low values of $radius$). The red dotted curve corresponds to performances of MtG (that does not depend on $d$ and $radius$): its \mschange{amount of \mesure} is around $1.9$\,KB.
Fig.~\ref{fig:random_BMTG} corresponds to the same experiment for MTGv2. In the worst cases, the \mschange{amount of \mesure} is around $3$\,KB. \\

\begin{figure}[tb]
    \centering
    \includegraphics[scale=\expeFigScaling]{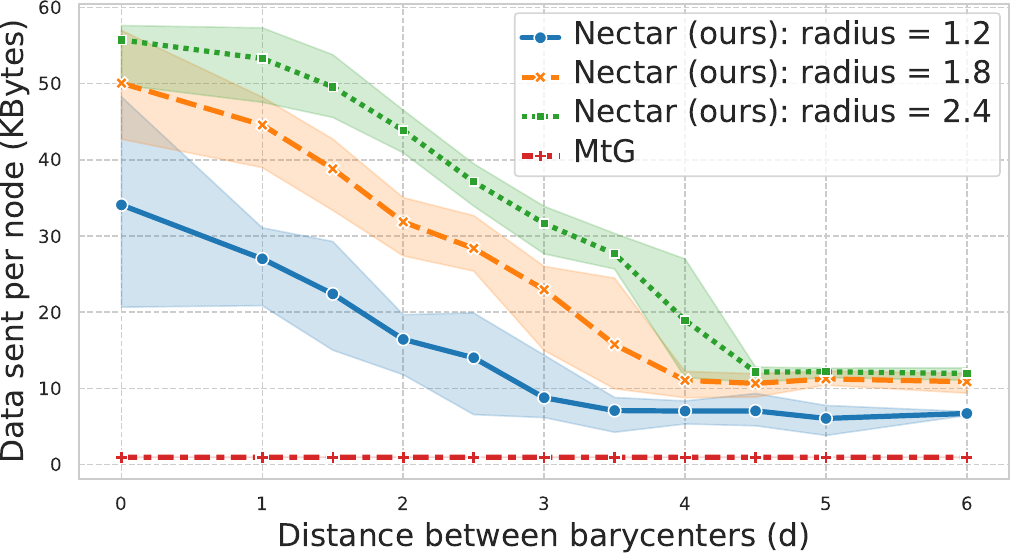}
    \caption{\ftchange{Data sent} per node (in KB), depending on the distance between barycenters ($d$), for different values of communication scope ($radius$), in the drone scenario for \solname. The red curve is MtG~\cite{bouget_mind_2018}, whose performance does not depend on $d$ nor $radius$.}
    \label{fig:random_BFT}
\end{figure}
\begin{figure}[tb]
    \centering
    \includegraphics[scale=\expeFigScaling]{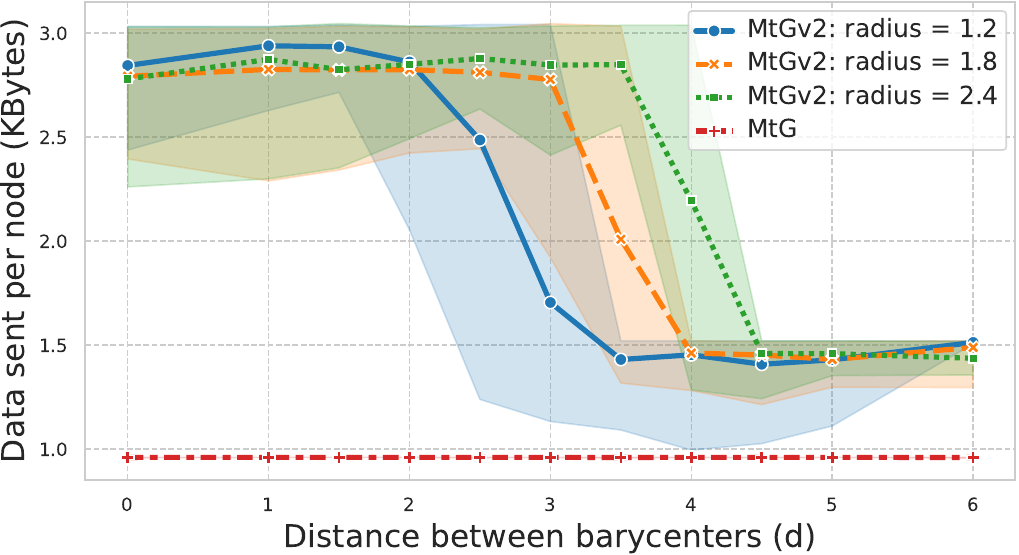}
    \caption{\mschange{Data sent} per node (in KB), depending on the distance between barycenters ($d$), for different values of communication scope ($radius$), in the drone scenario for MtGv2. }
    \label{fig:random_BMTG}
\end{figure}

\textbf{Number of nodes.}
The \mschange{network cost} per node in \solname increases in the worst cases quadratically with the number of nodes in the system and with the number of edges, as seen in Sec~\ref{sec:complexity}. To illustrate this effect, we now keep experimenting on regular and random graphs, but this time we vary the value of parameter $n$, for several values of $d$. The scope communication is fixed to $radius = 1.2$. 
Fig.~\ref{fig:random_n_BFT} shows the results for \solname. The maximum value observed is around $200$\,KB per node, for $n = 50$, with $d = 0$, corresponding to almost fully connected graphs. Fig.~\ref{fig:random_n_BMTG} shows the results for MTGv2. The maximum \mschange{amount of \mesure} per node observed is around $7.5$\,KB, for $n = 50$ and $d = 0$, also corresponding to almost fully connected graphs. \editms{Those values, both for \solname and MTGv2 algorithms, are very reasonable for nowadays technologies.} 

\begin{figure}[tb]
    \centering
    \includegraphics[scale=\expeFigScaling]{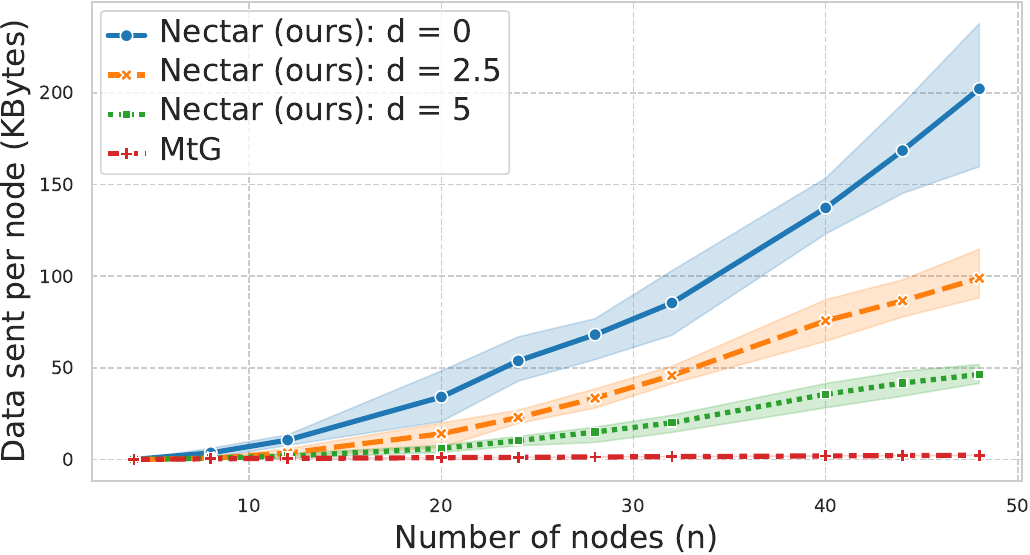}
    \caption{\mschange{\mesure } per node (in KB) depending on the number of nodes (n), for different values of distances (d) between barycenters, with a fixed communication scope (radius = 1.2), in the drone scenario.
    }
    \label{fig:random_n_BFT}
\end{figure}
\begin{figure}[tb]
    \centering
    \includegraphics[scale=\expeFigScaling]{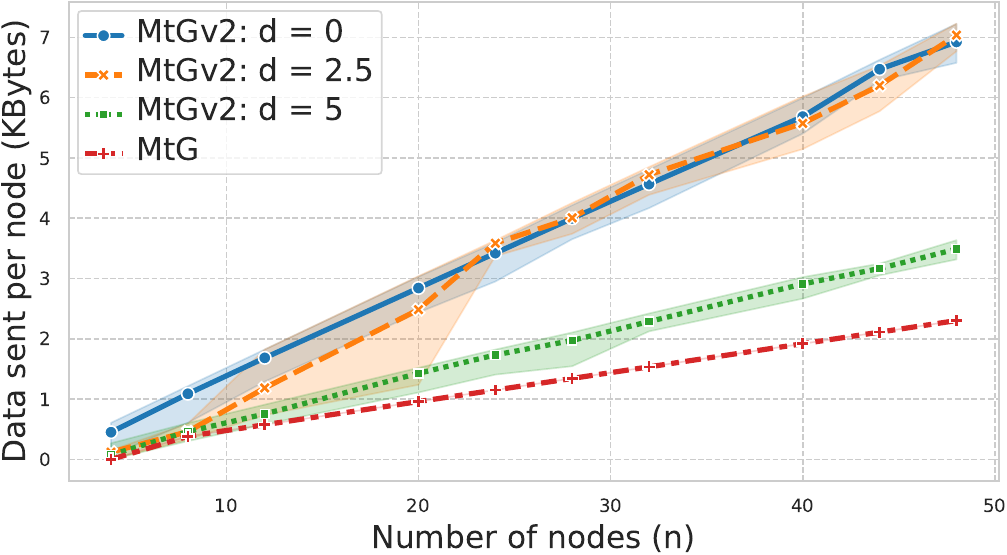}
    \caption{\mschange{\mesure } per node (in KB) depending on the number of nodes (n), for different values of distances (d) between barycenters, with a fixed communication scope (radius = 1.2), in in the drone scenario. The red line is the same as in Fig.~\ref{fig:random_n_BFT}.}
    \label{fig:random_n_BMTG}
\end{figure}

\subsection{Byzantine resilience} 
\label{sec:xp_byz}

\editms{\textbf{Drone based scenario.}} This section investigates the resilience of \solname, MtG, and MtGv2 to Byzantine behaviors.  \editms{We particularly study the robustness of MtG and MtGv2 protocols against few Byzantine nodes.} 
Due to the use of Bloom filters, MtG is easily corruptible by Byzantine nodes. For example, in a partitioned graph, Byzantine nodes can send filters full of $1$ values to lead correct nodes to conclude that the system is connected. We experiment the situation of a graph partitioned into two parts, and we take care of equally distributing the Byzantine nodes between the two parts. 
The red dotted curve of Fig.~\ref{fig:byz_succes} shows the proportion of correct nodes that correctly detect the partition. Such an experiment shows that two Byzantine nodes are enough to make all correct nodes reach the incorrect decision, while one Byzantine node is enough to prevent correct nodes from reaching the same decision (i.e., breaking the agreement property). 

This attack is not possible on \solname and MtGv2, due to the use of signatures (Byzantine nodes cannot forge signatures). Thus, we considered another attack, which is the same for the two tested algorithms.
We generated a subgraph of correct nodes that is partitioned into two parts. We then added Byzantine edges between each part, to make the graph connected, where all communications between the two correct parts must pass through Byzantine nodes, which also means that the graph is at most $t$-connected, where $t$ is the number of Byzantines nodes, and the Byzantine nodes are the $t$ key nodes that decide the connectivity parameter. The Byzantine behavior we considered for this kind of situation is that Byzantine nodes act correctly toward one part of the subgraph of correct nodes, and as crashed nodes for the other part. Fig.~\ref{fig:byz_succes} shows our results for these experiments. For \solname, as the connectivity will never be above $t$, all correct nodes conclude to a \partitionable decision, which is the correct decision since the subgraph of correct nodes is disconnected. For MtGv2, such an attack makes around half of the correct nodes conclude that the graph is connected (which is true), while the other half conclude that the graph is partitioned. Once again, one Byzantine node is enough to lead correct nodes to different decisions.  \\

\editms{\textbf{Connectivity-dependent topologies.} We investigated the same attacks on the topologies highlighted by Bonomi et al. \cite{bonomi2019multi}, with an aleatory placement of Byzantine nodes, and observed the following behaviors. For all topologies, MtG drops to 0 success rate of correct decision as soon as there are 2 Byzantine nodes, while \solname keeps a success rate of 1. For MtGv2, results depend on topologies: For $k-$diamond graphs, MtGv2 keeps a success rate close to 1 (with a confidence interval of [0.95, 1], no matter the number of Byzantine nodes. For $k-$regular graphs,  $k-$pasted graphs, Generalised Wheel graphs, and Multipartite Wheel graphs, MtGv2 drops to 0.3 success rate on average, with a confidence interval of [0, 1]. }

\editms{To conclude, this experimental evaluation shows that existing partition detection algorithms are easily corruptible by a few Byzantine nodes using simple attacks such as sharing incorrect information or omitting to send some messages, although strengthening solutions by using signatures complicates attacks. \solname, however, provides a solution that ensures Byzantine resilience, no matter the network's topology. Although this solution is more costly than state-of-the-art ones, this cost is lower than $500$\,KB per node per algorithm execution, up to 100 nodes, which is a \mschange{reasonable} \mschange{network cost} for most of nowadays networks.  }

\begin{figure}[tb]
    \centering
    \includegraphics[scale=\expeFigScaling]{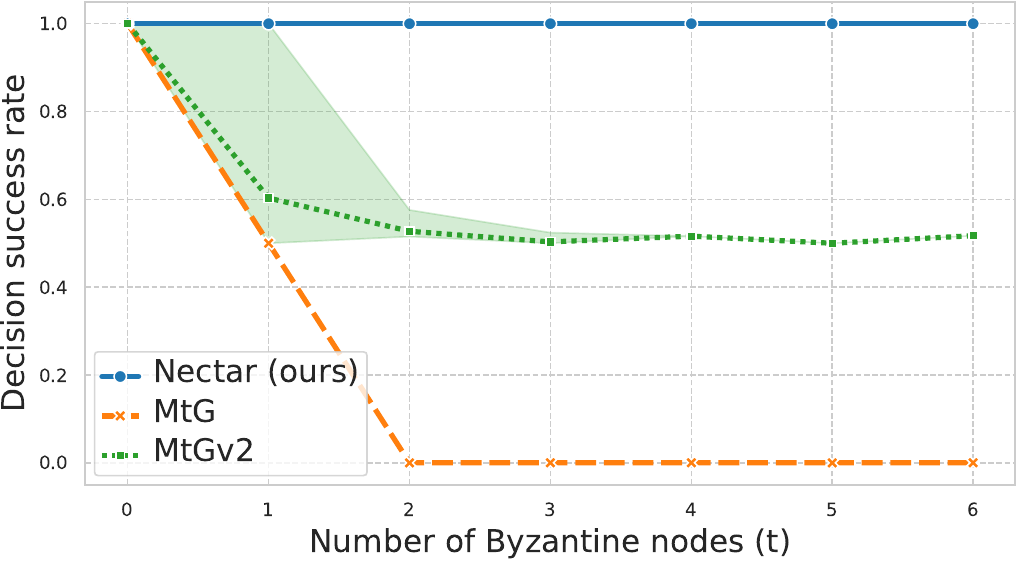}
    \caption{Success rate of correct decision (through correct nodes), depending on the number of Byzantine nodes, for the drone scenario with 35 nodes. Results with  20 and 50 nodes exhibit the same tendencies.}
    \label{fig:byz_succes} 
\end{figure}

\section{Related Work}
\label{sec:sota}

\subsection{Partition detection}
\label{netpartdetc}

 Several works have studied how to detect network partition in fault-free and in crash-prone networks, in particular in mobile ad hoc networks (MANETs) and Wireless Sensor Networks (WSNs). However, to the best of our knowledge, only \cite{augustine2022byzantine} considered Byzantine faults. 

\subsubsection{In fault-free networks}

    Ritter et al.~\cite{ritter_partition_2004} proposed a heuristic to detect partitions in reliable MANETs. Their approach is based on the hypothesis that nodes at the border of the graph do not change frequently and are reliably detectable. Their idea is to exchange beacon messages through the network. A partition is suspected when beacons are not received for sufficiently long between well-chosen nodes.  

    Bouget et al.~\cite{bouget_mind_2018} propose MindTheGap (MtG), a light and fully decentralized approach dedicated to MANETs. MtG assumes that all nodes are correct but tolerates unreliable communication channels and detects partitions in dynamic networks. Every node in MtG gossips what it knows of the system using Bloom filters. Simulations revealed that MtG detects 90\% of partitions despite a 40\% message loss rate.  \\
  
\subsubsection{In crash-tolerant networks}

Renesse et al.~\cite{renesse_gossip-style_1998} propose a crash failure detection service based on gossip that can be used to detect partitions in asynchronous networks. Due to asynchrony, it is hard to know if a process (a node) has crashed or is just very slow. This is why, in such a context, allowing some false positive detections (assuming a slow process can be crashed) is reasonable while respecting acceptable accuracy.  This failure detection algorithm is based on probabilistic properties and comes with guarantees about low rates of false detections, message-loss resilience, known probability of mistakes, and scalability.  The authors turn their failure detection algorithm into a partition detection algorithm by detecting a group of unreachable nodes. 
            
Conan et al.~\cite{conan_failure_2008} propose an approach of partition detection based on heartbeat vector and propagation of reachability information for asynchronous distributed systems. Reachability information is obtained by propagating and aggregating path information for every node. Their method tolerates node mobility by frequently rebuilding the topology of the reachable system. An effort is made to distinguish disconnections from crashes. It is why, while detecting an unreachable node, computations are made to evaluate if the node is potentially just unreachable or crashed by paying attention to previous cases of potential disconnection of this node.
            
Those partition detection approaches have been designed for correct and crash-tolerant networks, and would not perform adequately under the Byzantine fault model that we consider. 

\editms{Augustine et al.~\cite{augustine2022byzantine} address the problem of detecting if a graph is connected, in the context of Byzantine networks. While focusing on the very specific case of congested cliques, they provide an algorithm that detects if the subgraph of correct nodes is connected or if it is \textit{far from connected}, i.e., if it has at least $2t + 1$ connected component. The algorithm is executed in $O(polylog(n))$ time.}
    
\subsection{Reliable communication primitives}


Dolev investigated the question of reliable communication in partially connected networks in the presence of Byzantine nodes~\cite{dolev_unanimity_1981}. He proved that agreement is possible if and only if $t < k/2$ and $t < n/3$ where $t$ is the number of faulty nodes, $k$ is the connectivity of the graph, and $n$ is the number of nodes in the graph. While the $1/3$ of Byzantine nodes limit was well-known for fully connected networks~\cite{lamport_byzantine_1982, bracha_asynchronous_1987}, this work identified the partial connectivity requirement for the first time.  
Dolev also introduced an algorithm that provides reliable communication, for $(2t + 1)$-connected networks. The main idea of this algorithm is that traveling messages contain information about the path they followed in the network. By looking for that information, the nodes are able to compute the number of disjoint paths a message has traveled through. Nodes deliver a message when they are able to deduce that the message has traveled through $(t+1)$ disjoint paths. Two variants of this protocol are proposed, to deal with both known and unknown topologies. The message complexity of this protocol is high and, in the worst case, equal to $\mathcal{O}(n!)$ where $n$ is the number of nodes.

This reliable communication protocol combined with Bracha's reliable broadcast algorithm~\cite{bracha_asynchronous_1987} (introduced for fully connected networks) provides a reliable broadcast protocol for partially connected networks. However, since those two protocols are costly, some optimizations have been proposed in the state-of-the-art. Bonomi et al.~\cite{bonomi_practical_2021} optimized the combination of Bracha's and Dolev's protocols through protocol-specific and cross-layer optimizations. Simulations showed that such optimizations made the protocols more practical by reducing its latency and the amount of transmitted information. 
            
A new broadcast variant has been proposed by Khan et al.~\cite{khan_exact_2019} that suits well partially connected networks: \textit{local broadcast}. In the local broadcast context, nodes cannot send different values to their neighbors, even faulty nodes. This context makes sense in wireless communication, among others. Such a new primitive reduces the possible faulty behaviors of Byzantine nodes. Therefore, when local broadcast is possible, Khan et al. have shown that consensus primitives built on top of a local broadcast primitive have less restrictive limits than those built on point-to-point channels. Those limits are: $\lfloor 3t /2 \rfloor < k  $ and $2t < d $, where $k$ is the graph connectivity, $t$ is the number of Byzantine nodes in the system, and $d$ is the minimum node degree. Note that the constraint about graph connectivity is less restrictive. On the other hand, the restricted number of acceptable Byzantine nodes has been replaced by a condition about nodes degree.
Recent works provided real-time guarantees for reliable broadcast in fully-connected networks with probabilistic message losses~\cite{kozhaya2018rt, kozhaya2021pistis}. We consider it future work to extend those guarantees to the partially connected networks we consider in this paper. 
            
\subsection{Vertex-connectivity in distributed systems}    

DECK is an algorithm that computes a network's connectivity~\cite{akram_deck_2018}. It is fully decentralized and designed for asynchronous WSNets. However, it assumes that all nodes are correct when we consider Byzantine faults. Tucci Piergiovanni and Baldoni~\cite{tucci_2007_connectivity} proposed an adaption of the definition of vertex-connectivity for eventually quiescent dynamic distributed systems. Their work focuses, in particular, on the definition of strong connectivity in such systems, for which they propose an algorithm that computes a tree as an overlay topology that guarantees eventual connectivity.

\mschange{\subsection{Unknown systems and Detectors }

Distributed unknown systems assume that nodes know neither the number of nodes in the systems nor their identity. 
Cavin et al. \cite{cavin_consensus_2004} first proposed a consensus algorithm \ftchange{for this type of systems} based on a \textit{Participant detector} oracle. They however suppose a reliable (and asynchronous) network. 
In unknown networks, Greve et al. \cite{greve_knowledge_2007} showed that consensus \ftchange{can} be \ftchange{provided} only if a node can establish (with the help of a \textit{participant detector} oracle)  that the network has a connectivity of at least $t+1$, where $t$ is the maximal number of crashes.
In Byzantine networks, Alchieri et al. \cite{alchieri_knowledge_2018} showed that Byzantine consensus can be achieved in unknown networks under the assumptions of synchronous communication, with sufficient connectivity in the network (connectivity of at least $2t + 1$, where $t$ is the number of faulty nodes). 
\ftchange{In these last works,} both crash and Byzantine cases rely on the \textit{participant detector} oracles (or variants) defined by Cavin et al \cite{cavin_consensus_2004}.

\ftchange{In} dynamic graphs, several works \cite{kihlstrom_byzantine_2003, awerbuch_-demand_2002, greve_time-free_2012} \ftchange{have} proposed a Byzantine failure detector. Byzantine failure detectors can be seen as oracles that detect when a Byzantine node deviates from its correct \ftchange{behavior}. For our problem, we consider that Byzantine detectors are out of the scope since Byzantine nodes can attack with a behavior undifferentiable than correct nodes. }

\section{Conclusion}
\label{sec:conclusion}

In this paper, we investigated the question of partition detection in Byzantine networks. We highlight the necessity to reconsider the notion of partition for those cases of networks and propose a new \textit{t-Byzantine-partitionable} notion. We characterize \textit{t-Byzantine-partitionable} networks, \editms{ relate this notion to the graph vertex-connectivity property,} and propose a new algorithm \solname that detects such network property. \mschange{\solname assumes synchronous communication and relies on signatures. We speculate that the problem becomes \ftchange{impossible to solve} in an asynchronous environment, \ftchange{as in this case arbitrarily late} messages \ftchange{become} indistinguishable from \ftchange{the effect of a} partition. Nevertheless, we posit that it can be accomplished without signatures \ftchange{in synchronous networks}, albeit at a significant cost.
} We formally prove \mschange{\solname} correctness and propose an in-depth experimental evaluation. Our experimental evaluation aims to compare \solname to \editms{a non Byzantine-tolerant performant baseline and its signature-based variant, for several realistic families of topologies.} 
Our results show that the \solname algorithm maintains a 100\% accuracy of the tested scenarios while the accuracy of the best competitor baselines decreases by at least 40\% as soon as one participant is Byzantine. Detecting {t-Byzantine-partitionability} in networks implies a higher \mschange{network cost} \editms{($O(n^4)$ message complexity in the worst cases, i.e., in fully connected graphs)} than efficient \editms{ state-of-the-art} partition detection algorithms. \solname's \mschange{network cost} increases with the number of nodes and decreases with the diameter of the networks, but it always remains lower than around 500\,KB per node for up to 100 nodes. 

\mschange{\section*{Acknowledgments}

    This work was partially supported by the French ANR project ByBloS (ANR-20-CE25-0002-01) devoted to the modular design of building blocks for large-scale Byzantine-tolerant applications, by the EU Horizon Europe Research and Innovation Programme under Grant No. 101073920 (TENSOR) and by the Ecole Normale Supérieure (ENS) of Rennes.}

\printbibliography

\end{document}